\documentclass[preprint,12pt]{elsarticle}

\usepackage{algorithm}
\usepackage[noend]{algpseudocode}
\usepackage{amssymb,amsmath,amsthm}
\usepackage{multirow}
\usepackage{cleveref}
\usepackage{placeins}

\newtheorem{theorem}{Theorem}
\newtheorem{lemma}{Lemma}

\newtheorem{definition}{Definition}

\DeclareMathOperator{\poly}{poly}
\DeclareMathOperator{\OPT}{OPT}
\DeclareMathOperator{\ALG}{ALG}
\newcommand{\size}[1]{\lvert #1 \rvert}

\journal{Discrete Applied Mathematics}
\begin{document}
\begin{frontmatter}
\title{A Massively Parallel Dynamic Programming for Approximate Rectangle Escape Problem}

\author[sharif,sepid]{Sepideh Aghamolaei\corref{cor1}}
\ead[sepid]{aghamolaei@ce.sharif.edu}
\cortext[cor1]{Corresponding author.}

\author[sharif,ghodsi]{Mohammad Ghodsi}
\ead[ghodsi]{ghodsi@sharif.edu}

\affiliation[sharif]{
organization={Department of Computer Engineering, Sharif University of Technology},
addressline={Azadi Ave.},
city={Tehran},
country={Iran}
}
\affiliation[ghodsi]{
organization={School of Computer Science, Institute for Research in Fundamental Sciences (IPM)},
city={Tehran},
country={Iran}
}

\begin{abstract}
Sublinear time complexity is required by the massively parallel computation~(MPC) model.
Breaking dynamic programs into a set of sparse dynamic programs that can be divided, solved, and merged in sublinear time.

The rectangle escape problem (REP) is defined as follows:
For $n$ axis-aligned rectangles inside an axis-aligned bounding box~$B$, extend each rectangle in only one of the four directions: up, down, left, or right until it reaches $B$ and the density~$k$ is minimized, where $k$ is the maximum number of extensions of rectangles to the boundary that pass through a point inside bounding box $B$.
REP is NP-hard for $k>1$.
If the rectangles are points of a grid (or unit squares of a grid), the problem is called the square escape problem (SEP) and it is still NP-hard.

We give a $2$-approximation algorithm for SEP with $k\geq2$ with time complexity $O(n^{3/2}k^2)$. This improves the time complexity of existing algorithms which are at least quadratic. Also, the approximation ratio of our algorithm for $k\geq 3$ is $3/2$ which is tight.
We also give a $8$-approximation algorithm for REP with time complexity $O(n\log n+nk)$ and give a MPC version of this algorithm for $k=O(1)$ which is the first parallel algorithm for this problem.
\end{abstract}

\begin{keyword}
Dynamic programming \sep Massively parallel computation \sep Approximation algorithm
\end{keyword}

\end{frontmatter}
\section{Introduction}\label{sec:introduction}
Having sublinear time complexity and keeping the total time complexities of parallel machines is often a requirement (See \Cref{lit:mpc}).
Sparse dynamic programs~\cite{eppstein1992sparse,eppstein1992sparse2} compute a subset of the terms if only one specific term of the sequence is needed, which helps improve the time complexity.
Lower bounds based on the strong exponential time hypothesis (SETH) often do not hold for sparse dynamic programs, for example, see the algorithm of~\cite{bringman2018multivariate} and the lower bound in~\cite{bringmann2015quadratic,abboud2015quadratic}.
Also, designing a parallel dynamic programming (DP) algorithm by assigning a thread to each cell of the DP table and using a scheduler to assign these threads to processors (usually in blocks) in the order of filling the table often takes at least linear time. 
Several methods of making special classes of dynamic programs parallel have been introduced over the years~\cite{im2017efficient,aghamolaei2018geometric} (See \Cref{lit:dpmpc} for more details).

We use the following problems to demonstrate our method of making dynamic programs parallel:
\begin{itemize}
\item
In the rectangle escape problem (REP), a set of axis-aligned rectangles $R_1,R_2,\ldots,R_n$ are given inside a larger axis-aligned rectangle called the boundary $B$ and the goal is to extend the rectangles towards one of the boundaries such that the maximum number of extended rectangles that path over a point inside $B$ called the density is minimized. See \Cref{lit:rep} for a brief overview of the existing results on sequential algorithms for REP.
\item
The special case of REP where instead of rectangles we have a set of points from a grid and they called this problem square escape problem (SEP).
Roy, et al~\cite{fixed} introduced SEP and proved it is NP-hard even for $k=2$.
\end{itemize}
Ma, et al~\cite{ma} introduced the rectangle escape problem (REP) and proved REP is NP-hard for $k\geq 3$. Later, Assadi, et al~\cite{zar1} proved that it is NP-hard for $k\geq 2$ and gave a lower bound $3/2$ for its approximation factor.

Some dynamic programs can be made parallel by partitioning the entries of the table such that in each subset, the entries are all equal.
We give an example of the decomposition of dynamic programs with small integer solutions into subproblems with unit solutions. The resulting instances have the same table size, but they are much sparser than the original instance of the problem.

%%%%%%%%%%%%%%%
%\subsection{Rectangle Escape Problem (REP)}

\subsection{Contributions}
We give a $1+\frac{1}{k-1}$-approximation algorithm for SEP with disjoint points for $k\geq 2$ that has subquadratic time complexity. Our algorithm is not based on linear programming, so, it can be made parallel.
The same algorithm has approximation ratio $4$ for the (non-disjoint) square escape problem.

Our main contribution is the first parallel algorithm for this problem since linear programming has no polylogarithmic time parallel algorithm assuming $NC \neq P$.
We prove our algorithm also works in MapReduce models MPC~\cite{mpc} and MRC~\cite{mrc} for disjoint rectangles. Our algorithm on such inputs takes $O(1)$ rounds and $O(n\log n+kn)$ communications, which is in MPC.

The disjoint version of SEP is when each grid point appears at most once in the input.
\Cref{table:results} summarizes the results on the rectangle escape problem for $k\geq 2$.
\begin{table*}[h]
\centering
\begin{tabular}{|c|c|c|p{3.06cm}|p{2.35cm}|}
\hline
Range of $k$ & Approx. & Time & Reference & Problem\\
\hline
$k\geq 3$ &  $4$ & $O(n^2+T(n^2))$  & \cite{ma} & REP\\
$k\geq 2$ & $4$ & $O(n^2+T(n^2))$  & \cite{zar1} & REP\\
$k\geq 2$ & $\geq 3/2$ & $O(\poly(n))$  & \cite{zar1} & REP\\
%\hline
$k\geq 2$ & $4$ & $O(n^{3/2}k^2)$ & \Cref{alg:approx} & SEP\\
$k\geq 2$ & $1+\frac{1}{k-1}$ & $O(n^{3/2}k^2)$ & \Cref{alg:approx} & Disjoint SEP\\
$k\geq 2$ & $8$ & $O(nk)$ & \Cref{alg:approx2} & REP\\
$k\geq \frac{36}{\epsilon^2}\ln n$ & $1+\epsilon$ & $O(n^2+T(n^2))$ & \Cref{theorem:corrected} + \cite{zar1} & REP\\
\hline
\end{tabular}
\caption{Summary of the results on REP and SEP. Here, $T(n)$ denotes the time complexity of solving a linear program with $O(n)$ constraints on $n$ variables.}
\label{table:results}
\end{table*}
\section{Preliminaries}

\subsection{Massively Parallel Computations and MapReduce Model}\label{lit:mpc}
MapReduce is a parallel and distributed framework for processing large data sets. In this framework, data is distributed among a set of machines, which process their data independently in parallel during each round. At the end of each round, the machines communicate with each other.
Two of the theoretical models for MapReduce are MapReduce Class (MRC)~\cite{mrc} and MPC~\cite{mpc}. In the MRC model, the constraints for an input of size $n$ are:
\begin{itemize}
\item the number of machines $(L)$ satisfies $L=O(n^{\eta})$, for a $\eta\in(0,1)$,
\item the memory of each machine $(m)$ satisfies $m=O(n^{\eta'})$, for a $\eta'\in(0,1)$,
\item the number of rounds $(t)$ satisfies $t=O(\log^{\eta''} n)$, for a $\eta''\geq 0$.
\end{itemize}
In Massively Parallel Computations (MPC), all MRC conditions must hold, also the total communication of each round must be linear, i.e. $mL=O(n)$, and the replication factor be constant.

Parallel algorithms with sublinear memory in each machine are usually used to model the MapReduce framework~\cite{mrc,mpc}. It has been shown that CRCW PRAM algorithms work in MapReduce with the same asymptotic time/round complexity~\cite{goodrich2011sortingI}.
Round compression implements distributed algorithms in MapReduce using a fewer number of rounds~\cite{czumaj2019round}.

\subsection{Dynamic Programming Algorithms in MPC}\label{lit:dpmpc}
There is a method for solving dynamic programming with approximation factor $1+\epsilon$ in MapReduce assuming the dynamic program has monotonicity and decomposability properties~\cite{im2017efficient}.
The definition of these properties is as follows~\cite{im2017efficient}:
\begin{itemize}
\item Monotonicity in a minimization problem means the solution to subproblems must be smaller than or equal to the objective function.
\item Decomposability is when the input can be divided into a two-level laminar family (upper level and lower level) such that an approximate solution to the problem can be constructed by concatenating the solutions to the upper level, and an approximate solution to the upper level can be constructed using the solutions from $O(1)$ lower level sets.
\end{itemize}
This method requires $O(\log_m n)$ rounds and $O(n)$ communications.

An exact dynamic programming algorithm in MapReduce with the following conditions: sparsity, neighbors, order-preserving, parallelizable, and summarizable is called a mergeable dynamic program and can be solved in $O(\log_m n)$ rounds~\cite{aghamolaei2018geometric}. The definition of these concepts is as follows~\cite{aghamolaei2018geometric}:
\begin{itemize}
\item The sparsity condition says that the number of non-empty cells of the table must be at most linear in the input size.
\item Neighbors condition restricts the subproblems to depend on at most $O(1)$ other subproblems.
\item Order preserving means that the value of a cell $DP[i][j]$ can only depend on the values of cells $DP[i'][j'], i'\leq i, j'\leq j, \phi(i',j') \leq \phi(i,j)$, where $DP$ is the dynamic programming table and $\phi$ is a valid order of computing the values of the dynamic programming table.
\item Parallelizable condition restricts the recurrence relation used for computing the values of the dynamic programming table to semi-group functions. Semigroup functions are associative binary functions such as minimum, union, and summation.
\item Summarizable means for all cells $(i,j)$ in a subproblem $s$, their value depends on subproblems with cells $(i',j')$ such that $\phi(i,j)-\phi(i',j') \leq \ell$, for a known sublinear $\ell$.
\end{itemize}
This algorithm takes $O(\log_m n)$ rounds and $O(n\log_m n)$ communications \cite{aghamolaei2018geometric}.

\subsection{Maximum Bipartite Matching}
There is an algorithm with $O(\sqrt{\size{V}}\size{E})$ for maximum matching in a graph with the set of vertices $V$ and the set of edges $E$~\cite{cormen2022introduction}. The best known parallel algorithm for this problem is in $\mathrm{RNC}^2$~\cite{mulmuley1987matching,karp1985constructing} and no $\mathrm{NC}$ algorithms are known for the problem.
%%%%%%%%%%%%%%%%%%%%%
\subsection{The Linear Programming Algorithm for Approximating REP}\label{lit:rep}
Here, we review the integer linear program of REP and $4$-approximation algorithm of Ma, et al~\cite{ma}.  For more on the problem, see \ref{app:rep}.
\begin{definition}[Escape path of a rectangle]
The rectangle formed by extending an input rectangle to the boundary in one of the 4 directions (up, down, left, and right) is an escape path. Each rectangle must have at least one escape path in the optimal solution of REP. ($p_{i,\alpha}, \alpha = \text{right, left, up or down}$)
\end{definition}
\begin{definition}[Rectangle escape grid]\label{def:grid}
For a set of rectangles, the grid $G$ formed by extending the edges of rectangles to the boundary is the grid of the rectangle escape problem. This grid has size $O(n^2)$ since each rectangle adds at most $2$ vertical and $2$ horizontal lines to the grid.
\end{definition}

\begin{lemma}[\cite{ma}]\label{lemma:prev}
\Cref{rounding} rounds the fractional solution of the relaxed linear program of \ref{eq:LP1} to get a $4$-approximation for REP, where $C$ is the set of all rectangle escape grid's cells and $p_{i,\alpha}$ is the path of $i$-th rectangle in direction $\alpha$. $r_{i,\alpha}$ is the indicator variable of rectangle $i$ escaping in the direction $\alpha$.
\end{lemma}
\begin{equation}\tag{LP.1}\label{eq:LP1}
\begin{array}{l l l}
\text{minimize }& k&\\
\text{subject to}&&\\
&1 \leq r_{i,l}+r_{i,r}+r_{i,u}+r_{i,d} & 1\leq i \leq n \\
&\displaystyle\sum_{ p_{i,\alpha}\ni c}  r_{i,\alpha} \leq k & \forall c\in C \\
&r_{i,\alpha} \leq 1 &\substack{ 1\leq i \leq n\\ \alpha \in \lbrace r,l,u,d \rbrace} \\
& r_{i,\alpha}, k \geq 0 &\substack{ 1 \leq i \leq n\\ \alpha \in \lbrace r,l,u,d \rbrace }\*
\end{array}
\end{equation}
\begin{algorithm}
\caption{Deterministic rounding of \ref{eq:LP1}~\cite{ma}}
\label{rounding}
\begin{algorithmic}[1]
\Require \text{Rectangles $R_i, 1\leq i \leq n$, Boundary $B$}
\Ensure \text{An escape path for each rectangle $R_i$}
\State Run \ref{eq:LP1} to find the optimal fractional solution ($r^*_{i,\alpha}$).
\For {$1\leq i \leq n$}
	\State  $\beta \gets \arg\max_{\alpha \in \lbrace l,r,u,d \rbrace} r^*_{i,\alpha}$
	\State $r_{i,\beta}\gets 1$, $r_{i,\alpha}\gets 0, \alpha\in \lbrace l,r,u,d \rbrace \setminus \lbrace \beta \rbrace$
\EndFor
\State \Return the escape path with $r_{i,\alpha}=1; \alpha \in \lbrace l,r,u,d \rbrace$ for $R_i$.
\end{algorithmic}
\end{algorithm}
\subsection{More Definitions About REP}
We require some new definitions in this paper, which we discuss here.

\begin{definition}[Rectangle Level]\label{def:rlevel}
The minimum density of a rectangle extending to the boundary in one of the $4$ directions in a REP solution is called the level of that rectangle.
\end{definition}

\begin{definition}[Boundary density]
The density of SEP considering only the boundary points is called the boundary density and is denoted by $k_B$.
\end{definition}

Let $\bar{\alpha}$ be the opposite direction of $\alpha$ and let direction $\perp \alpha$ be perpendicular to $\alpha$. For example, if $\alpha=$ right, then, $\bar{\alpha}=$ left.
\begin{definition}[Escape DAG of REP]
We define the escape DAG of a REP instance with disjoint rectangles for direction $\alpha$ as the directed acyclic graph (DAG) $T_{\alpha}$ with the set of rectangles as its vertices and there is an edge from a rectangle $r_i$ to a rectangle $r_j$ if $r_i$ has to escape before $r_j$ in direction $\bar{\alpha}$.
\end{definition}
%%%%%%%%%%%%%%%%%%%%%%%%%
\section{A Parallel Approximation Algorithm for SEP}
First, we discuss the overview of our technique for approximate parallel dynamic programs based on sparse dynamic programs. Then, we show how to use it to find a parallel approximation algorithm for REP with $k\geq 2$. 

\subsection{Overview of Parallel Sparse Dynamic Programs}

A dynamic program with a bounded and increasing integer objective function $DP[(v_1,v_2,\ldots,v_m)]=I$ can be converted into a binary version $BDP[(v_1,v_2,\ldots,v_m,I)]=0/1$, which indicates if the objective function has increased compared to the terms used to compute it.

The sparse dynamic program of the binary dynamic program (BDP) finds the events that give the border between $BDP[(v_1,v_2,\ldots,v_m,I)]=1$ and $BDP[(v_1,v_2,\ldots,v_m,I)]=0$, then, these boundaries are used to avoid computing useless cells of the table.

Some of these boundaries might be simplified to be faster to compute at the cost of finding approximations of the objective function. To make sparse dynamic programs parallel, we focus on reducing the number of boundaries and computing the values for each unit of the resulting approximate objective function at each level.

\subsection{Applying The Technique on REP}
The dynamic program for the square escape problem based on the escape grid (\Cref{def:grid}) is:
\[
dp[i][j]=\min(dp[i-1][j],dp[i][j-1],dp[i+1][j],dp[i][j+1]),
\]
and the order of filling the dynamic programming table $(\phi)$ is in the order of increasing levels.
Let $M$ be the binary matrix that shows if a rectangle is at cell $[i][j]$ of the rectangle escape grid:
\[
M[i][j]=\begin{cases}
1 & \text{if there is a square at cell }(i,j)\\
0 & \text{otherwise}
\end{cases}
\]
and let the order of filling the dynamic programming table be
\[
\phi=\min(\sum_{t< j} M[i][t], \sum_{t< i} M[t][j],\sum_{t> j} M[i][t], \sum_{t> i} M[t][j]).
\]

The peeling order $\phi$ is defined by removing the unit density points and recomputing the densities and repeating this process. This order guarantees in no direction, a point is seen after the points that are in its path to the boundary. Since the points in the levels $1$ to $i$ of the peeling are a superset of the points that can escape with density at most $i$, this is a valid ordering to fill the dynamic programming table.

We prove each level can escape with approximation factor $2$ and each level adds either one or two to the density. So, we compute the levels and solve each level separately.

The length of the longest path in the execution graph of \Cref{alg:approx2} is $k$, so, the lower bound on the time complexity of parallel versions of that algorithm is $k$ (the lower bound does not hold for MPC). So, the best parallel version of the algorithm we can hope for takes at least $\Omega(k)$ rounds and in this section, we give an algorithm with $O(k)$ rounds (\Cref{alg:grid}).

\subsection{A Near-Linear Time Approximation Algorithm for Disjoint REP}
Here, we give a $4$-approximation algorithm for REP with $O(n\log n+nk)$ time, assuming $k\geq 2$ and the input rectangles are disjoint, using a new concept that we call the levels of a rectangle escape problem.

To compute $T_{\alpha}$, sweep in the direction $\alpha$ and keep an interval tree on the projection of the rectangle in direction $\perp \alpha$ as well as a balanced binary search tree (BST) in direction $\alpha$ storing where each rectangle ends. Use the BST to remove the rectangles from the interval tree (and the BST itself) after the sweep line passes them. The edges are computed when a segment (rectangle) is added and the interval tree is queried to find the set of intersecting intervals.

The time complexity of building this DAG is $O(n\log n)$, as each rectangle is queried in the interval tree at most $O(n)$ times (once for the coordinates of the start and end of each rectangle) and it is at most added and deleted once, each of which takes $O(\log n)$ time.

\begin{algorithm}[h]
\caption{Approximate Rectangle Escape Problem}
\label{alg:approx2}
\begin{algorithmic}[1]
\Require \text{a set of disjoint rectangles $R$, Boundary $B$}
\Ensure \text{an escape path for each rectangle in $R$}
\State{$\rho\gets 0$}
\State{$H\gets \emptyset$}
\For{each direction $\alpha=\{$left, right, up and down$\}$}
\State{Sweep $R$ in direction $\alpha$ and build an escape DAG $T_{\bar{\alpha}}$ in direction $\bar{\alpha}$.}
\State{Topologically sort $T_{\bar{\alpha}}$.}
\EndFor
\While{$H\neq R$}
\For{each direction $\alpha=\{$left, right, up and down$\}$}
\State{$L_{\alpha}=$ the rectangles with indegree $0$ in $T_{\alpha}$.}
\State{$H=H\cup L_{\alpha}$}
\For{$r\in L_{\alpha}$}
\State{$W[r]=\alpha$}
\EndFor
\State{Remove $L_{\alpha}$ from $T_{\alpha}$.}
\EndFor
\State{$\rho\gets \rho+1$}
\EndWhile
\\ \Return{$W[r]$, for $r\in R$}
\end{algorithmic}
\end{algorithm}

We give a peeling algorithm that removes the rectangles that can escape without going through another rectangle and continues until no more rectangles are left.
An example of our peeling algorithm is shown in \Cref{fig:peeling}.
\begin{figure}[h]
\centering
\includegraphics[scale=0.8]{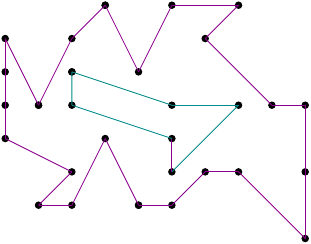}
\caption{Peeling points based on visibility to the axis-aligned bounding box. There are two levels. It has two levels, shown in different colors.}
\label{fig:peeling}
\end{figure}

\begin{lemma}\label{theorem:f}
\Cref{alg:approx} is a $2$-approximation if there is only one level.
\end{lemma}
\begin{proof}
A $2$-approximation for a REP with one level is to route each rectangle in the direction that intersects with the minimum number of rectangles.
We show that there can be no other rectangle in the direction with minimum density since there is only $1$ level.

The paths of rectangles can intersect at most once since otherwise if another rectangle passes through that cell, its escape path intersects one of the previous rectangles, which contradicts the condition of having only $1$ level (See \Cref{fig:case1}).

\begin{figure}[h]
\centering
\includegraphics[scale=0.8]{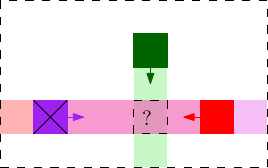}
\caption{Part of a REP instance with $1$ level. The purple (leftmost) rectangle cannot escape to the right and increase the density of the cell marked with `?', because of the red (rightmost) rectangle (See the proof of \Cref{theorem:f}).}
\label{fig:case1}
\end{figure}

Also, the path $P_r$ of a rectangle $r$ can only intersect with the path $P_{r'}$ of another rectangle $r'$, if $P_r$ and $P_{r'}$ are perpendicular, otherwise, the path of one rectangle would have to go through the other which contradicts the constraint that there is only $1$ level.

So, the escape density with the escape paths of \Cref{alg:approx} is the maximum of the densities of these cases, which is $2$.
\end{proof}

\begin{lemma}\label{lemma:add}
In a REP instance with $k\geq 2$, removing $4$ levels decreases the density by at least $1$ and at most $8$.
\end{lemma}
\begin{proof}
Let $\rho$ be the number of levels.
Any rectangle at level $5$ or more has to go through at least one rectangle in one of the $4$ first levels. So, its escape density is at least $1$. Removing the first (outer) $4$ levels gives an instance with density at most $k-1$.
Removing levels does not change the level of the remaining rectangles.
By induction on the number of levels divided by $4$, the density satisfies $k\geq \lceil \rho/4\rceil$.

Removing a level cannot decrease the density by more than $2$ as each level has density at most $2$ (\Cref{theorem:f}). So, $k\leq 2\rho$. Putting the two inequalities together gives:
\[
\lceil \frac{\rho}{4}\rceil \leq k \leq 8\lceil \frac{\rho}{4} \rceil.
\]
\end{proof}

\begin{theorem}\label{theorem:greedy}
\Cref{alg:approx2} is a $8$-approximation for REP, assuming $k\geq 2$.
\end{theorem}
\begin{proof}
\Cref{alg:approx2} routes each level independently, so, based on \Cref{theorem:f}, it has density at most $2$ for each level, resulting in density $2\rho$ in total.
Using \Cref{lemma:add} on density $(k)$ and the number of levels $(\rho)$ after removing the ceiling function, we have
\[
\frac{\rho}{4} \leq k \leq 8\frac{\rho+3}{4}+8 \Rightarrow \frac{k}{2}-3 \leq \rho \leq 4k.
\]
From these two inequalities, and knowing that the density cannot be smaller than $k$, we get the following inequality which gives us the approximation ratio:
\[
\max(k,k-6) \leq 2\rho \leq 8k \Rightarrow \frac{8k}{k} =8.
\]
\end{proof}

\begin{theorem}\label{theorem:time}
The running time of \Cref{alg:approx2} is $O(n\log n+kn)$.
\end{theorem}
\begin{proof}
Sorting the coordinates of the rectangle vertices takes $O(n\log n)$ time.
At each step of the sweeping algorithm, insertions and deletions in the BST to determine the next step, each of which takes $O(\log n)$ time and range queries, insertions, and deletions to the interval tree which also take $O(\log n)$ time per query. For each rectangle, a constant number of queries are made, so, the amortized time complexity of processing each rectangle is $O(\log n)$.
Sweeping $O(n)$ coordinates and building the DAG take $O(n\log n)$ time.

For each direction, we perform a topological sort on a DAG with $O(n)$ vertices and edges. This takes $O(n)$ time.

At each iteration of the algorithm, the set of rectangles that can escape with density one in a specific direction are the set of nodes of indegree $0$ in the DAG, which can be computed in $O(1)$ time per node, given the topological sort.
The number of iterations of the while loop is as most the number of levels.
Based on \Cref{lemma:add}, the number of levels is $O(k)$. Multiplying the number of rounds with the time complexity of each round gives $O(kn)$ time.
Summing up these time complexities gives the total time complexity of the algorithm which is $O(n\log n+kn)$.
\end{proof}

%%%%%%
\subsection{Computing The Peeling in MPC}

\begin{algorithm}[h]
\caption{A MPC Approximation Algorithm for SEP with $k=O(1)$}
\label{alg:grid}
\begin{algorithmic}[1]
\Require{A set of points $R$}
\Ensure{An escape path for each point of $R$}
\State{$S\gets \emptyset$}
\State{$Q\gets \emptyset$}
\State{$T\gets R$}
\While{$\size{T}\neq \emptyset$}
\State{$x_{m,y}\gets\min_{x',(x',y)\in T} x'$}
\State{$x_{M,y}\gets\max_{x',(x',y)\in T} x'$}
\State{$y_{m,x}\gets\min_{y', (x,y')\in T} y'$}
\State{$y_{M,x}\gets\max_{y', (x,y')\in T} y'$}
\For{$(x,y)\in T$ \textbf{in parallel}}
\If{$x=x_{m,y}$}
\State{$Q\gets Q\cup \{(x,y,\text{left})\}$}
\ElsIf{$x=x_{M,y}$}
\State{$Q\gets Q\cup \{(x,y,\text{right})\}$}
\ElsIf{$y=y_{x,m}$}
\State{$Q\gets Q\cup \{(x,y,\text{down})\}$}
\ElsIf{$y=y_{x,M}$}
\State{$Q\gets Q\cup \{(x,y,\text{up})\}$}
\EndIf
\EndFor
\State{$T\gets T\setminus \{(x,y)\mid \exists \alpha\in \{\text{left, right, up and down}\} (x,y,\alpha)\in Q\}$}
\EndWhile
\\ \Return{$Q$}
\end{algorithmic}
\end{algorithm}

\Cref{alg:grid} takes $O(k)$ rounds in MPC. The values computed in the first four lines of the while loop are required in the for loop, so, they need separate rounds. The minimums and maximums can be computed in parallel as they only require four copies of the points which can be computed in one round. Each iteration of the while loop takes $O(1)$ rounds and the number of iterations of the loop is at most $k$, which gives $O(k)$ rounds in total.

%\FloatBarrier
%
\section{A $2$-Approximation Algorithm for SEP}\label{sec:disjoint}
We give \Cref{alg:approx} for SEP and prove it is a $2$-approximation (more specifically, a $(1+\frac{1}{k-1})$-approximation) algorithm with subquadratic time complexity for $k\geq 2$.
The algorithm finds a bipartite matching between the input points (set $R$) and at most $k$ copies of the boundary points. Each edge of the matching shows an escape path in the solution computed by the algorithm.

\begin{algorithm}[h]
\caption{Approximate SEP}
\label{alg:approx}
\begin{algorithmic}[1]
\Require \text{a set of points $R$, Boundary $B$}
\Ensure \text{an escape path for each point of $R$, an integer $k_B$}
\For{$r\in R$ \textbf{in parallel}}
\State{$L_r=$ the perpendicular projections of $r$ on the boundary edges.}
\State{$E=E\cup \{r\}\times L_r$.}
\EndFor
\For{$k_B=1,\ldots,n$}
\State{$V=(\cup_{r\in R} L_r)\times\{1,2,\ldots,k_B\}$}
\State{$E'=\{(u,(v,i))\mid (u,v)\in E, i=1,2,\ldots,k_B\}$}
\State{Build a bipartite graph $G$ with vertex sets $R$ and $V$ and edges $E'$.}
\State{Compute a maximum matching $M$ in $G$.}
\If{$\size{M}=\size{R}$}
\State{Break (from the loop).}
\EndIf
\EndFor
\For{$(u,v)\in M$}
\State{$\alpha=$ the direction where the projection of $u$ is $v$.}
\State{$o_r=\alpha$.}
\EndFor
\\ \Return{$\cup_{r\in R} o_r$ and $k_B$}
\end{algorithmic}
\end{algorithm}

\begin{lemma}\label{lemma:boundary}
\Cref{alg:approx} finds the minimum boundary density.
\end{lemma}
\begin{proof}
\Cref{alg:approx} tests all values $k_B$ and reports the smallest density for which there are escape paths for all rectangles. So, it is enough to prove the decision based on matching in a bipartite graph works.

For each boundary vertex, $k_B$ copies of it are added to the second set of vertices of $G$. If there is a solution for SEP with boundary density $k_B$, the intersection of the escape path of a point $r\in R$ gives the boundary point $b$ for $r$ and $b$ is used at most $k_B$ times.
This means a matching edge $(r,b,i)$ exists for some $i\in \{1,2,\ldots,k_B\}$. The opposite is also true since each edge of the matching shows an escape path for each point and the boundary density is at most $k_B$ because there are at most $k_B$ edges for each boundary point. So, a matching in $G$ for a specific value $k_B$ in the for loop is a solution for SEP with boundary density $k_B$.
\end{proof}

The disjoint version of SEP is when each grid point appears at most once in the input.
\begin{lemma}\label{lemma:approx}
The density of disjoint SEP at the boundary is at least half the density, i.e., $k_B\geq k-1$, for $k\geq 2$.
\end{lemma}
\begin{proof}
If the maximum density happens at the boundary the problem is solved $(k_B=k)$. Otherwise, let $c$ be the internal point that has the maximum density $(k)$. Using the assumption that the points are disjoint, then, at most one input point (from set $R$) can be on $c$. So, at most one escape path originates from $c$.

Using induction on the distance (number of points in the grid) between points of density $k-1$ to the boundary, we show that a boundary cell with density at least $k-1$ exists. The base case is $c$ with density $k$. Now we discuss the induction step.
For a path to go through a point of density $k-1$ in the optimal solution, its other paths must be blocked by density cells of density at least $k-1$ otherwise escaping in a different direction would have reduced the density which contradicts the optimality of the solution. This gives a point of density at least $k-1$ which is closer to the boundary. So, $k_B \geq k-1$.
\end{proof}

\begin{lemma}\label{lemma:approx2}
The density of SEP at the boundary $(k_B)$ is at least $k/4$.
\end{lemma}
\begin{proof}
The points can escape in four directions, so, at least $k/4$ of them escape in the same direction. These points create density $k/4$ at the boundary.
\end{proof}

\begin{theorem}
\Cref{alg:approx} is a $2$-approximation for disjoint SEP and a $4$-approximation for SEP with $k\geq 2$.
\end{theorem}
\begin{proof}
First we discuss disjoint SEP.
\Cref{lemma:boundary} proves $k_B$ is minimized and \Cref{lemma:approx} proves $k_B$ is at least $k-1$. We know that $k_B\leq k$, as $k$ is the maximum density over all grid points, including the boundary points. So,
\[
k_B \leq k, \; k_B\geq k-1 \Rightarrow k_B \leq k \leq k_B+1\Rightarrow \frac{k_B+1}{k_B}\leq 1+\frac{1}{k-1}.
\]
This gives approximation ratio $2$, for $k\geq 2$.

For SEP, the bound is $k/4$ instead, which gives the same approximation ratio as before:
\[
\frac{k}{4} \leq k_B,\; k_B \leq k \Rightarrow k_B \leq k \leq 4 k_B \Rightarrow \frac{4k_B}{k_B}=4.
\]
\end{proof}

\begin{theorem}
\Cref{alg:approx} takes $O(n^{3/2}k^2)$ time.
\end{theorem}
\begin{proof}
In two sets of vertices of $G$ there are $\size{V}=n$ and $n\leq \size{U}\leq 4n$ vertices. There is an edge between each point and $4$ boundary points, each of which is repeated $k$ times, resulting in $4k$ edges per input point. So, $\size{E}=4kn$.
Computing a bipartite matching takes $O(\sqrt{\size{V}+\size{R}}\size{E})=O(kn\sqrt{n})$ time. The loop repeats $k_B$ times, so, the algorithm takes $O(n^{3/2}k^2)$ time.
\end{proof}

\section{Analysis of Randomized Rounding for REP}
The $(1+\epsilon)$-approximation algorithm in~\cite{zar1} uses an upper bound on the expected value instead of the expected value (or both a lower bound and an upper bound) in the definition of Chernoff bound. We fix this in~\Cref{theorem:corrected} which only affects the constant in the lower bound on the values $k$ for which the algorithm gives a $(1+\epsilon)$-approximation solution.
\begin{theorem}\label{theorem:corrected}
The randomized rounding of the solution of LP.1 gives a $(1+\epsilon)$-approximation for $k\geq \frac{36}{\epsilon^2} \ln n$ and any arbitrary constant $\epsilon\in (0,3)$, with high probability.
\end{theorem}
\begin{proof}
Consider a rectangle $R_i$ whose path goes through a grid cell $c$. If $c$ is inside $R_i$, then any escape direction of $R_i$ increases the density of $c$ by one. After removing these cells, the escape paths of $R_i$ become independent from each other.
The indicator variables of the changed escape paths of rectangles are $x_{i,\alpha}, i=1,\ldots,n, \alpha\in \{\text{right,left,top,bottom}\}$ in the random solution.
The density of a cell $c$,denoted by $D_c$, in the randomly perturbed solution is:
\[
D_c=\sum_{R_{i,\alpha}\ni c} x_{i,\alpha} \Rightarrow
E[D_c] =\sum_{R_{i,\alpha}\ni c} E[x_{i,\alpha}]
\]

Let $\OPT_f$ be the cost of the optimal fractional solution, $\OPT$ be the cost of the (integral) optimal solution, and $\ALG$ be the cost of the deterministic rounding of LP.1 (\Cref{rounding}). The algorithm rounds every variable $r_{i,\alpha}\geq 1/4$ to $1$ and the rest to $0$. So, the cost of the algorithm is at most $4$ times the cost of the fractional solution (\Cref{lemma:prev}):
\[
\OPT_f \leq \OPT \leq \ALG \leq 4\OPT_f \leq 4\OPT \Rightarrow \frac{\OPT}{4} \leq \OPT_f\leq \OPT.
\]
This gives us a lower bound on the fractional solution in terms of the integral solution $\OPT_f \geq \OPT/4$. Also, because of the randomized rounding, we have
\[
\max_{c} E[D_c] =\max_c \sum_{R_{i,\alpha}\ni c} E[x_{i,\alpha}]=\max_c \sum_{R_{i,\alpha}\ni c} r^*_{i,\alpha}=\OPT_f.
\]
So, the cost of the randomized rounding algorithm satisfies
\[
\max_{c} E[D_c]\geq k/4,
\]
where we used $\OPT=k$, too.

For each cell $c$, the variables $x_{i,\alpha}$ for $R_{i,\alpha}\ni c$ are independent, as the path of a rectangle that does not contain $c$ can only intersect $c$ in one direction and the path of a rectangle that contains $c$ always intersects with $c$ (so, in the density formula it can be replaced with a $1$).
Using Chernoff bound, $D_c$ bounds as follows:
\[
pr(D_c\geq (1+\epsilon)E[D_c]) \leq e^{-E[D_c]\epsilon^2/3}
\]
Applying the union bound gives:
\[
pr(\exists c: D_c \geq (1+\epsilon)E[D_c])
\leq \sum_c pr(D_c \geq (1+\epsilon)E[D_c])
\leq \sum_c e^{-E[D_c]\epsilon^2/3}.\]
So, the probability that the algorithm finds a solution with density $(1+\epsilon)k$ is at least:
\begin{align*}
1-pr(\exists c: D_c \geq (1+\epsilon)E[D_c])&\geq 1- \sum_c e^{-E[D_c]\epsilon^2/3}\\
&\geq 1-(2n)^2 e^{-\max_c E[D_c]\epsilon^2/3}\\
&\geq1- 4n^2 e^{-(k/4)\epsilon^2/3},
\end{align*}
where we used $\max_c E[D_c]\geq k/4$ in the last inequality.

Substituting $k\geq \frac{36}{\epsilon^2} \ln n$, then, for $n\geq 4$, we have
$
1- 4n^2 e^{-(k/4)\epsilon^2/3}\geq 1-\frac{4}{n}.
$

For $n=1$, any solution is optimal. For $n\geq 2$ and $\epsilon\in (0,3)$, we have
$k\geq \lceil \frac{36}{\epsilon^2} \ln n\rceil \geq \lceil4\ln 2\rceil = 3$. Using $n\geq k$ and the same inequality again, $k\geq \lceil \frac{36}{\epsilon^2} \ln n\rceil \geq \lceil4\ln 3\rceil = 5$.
Since $n\geq k$, we have $n\geq 5$. So, $\frac{4}{n}\leq \frac{4}{5}$ and the probability is always positive and we do not need to add any other restrictions on $n$ to guarantee the success probability $1-\Omega(1/n)$.
\end{proof}

\section{Conclusions}
We used sparse dynamic programming on an approximation algorithm based on linear programming to find a MPC algorithm for approximate rectangle escape problem (REP) with disjoint rectangles.
Improving the approximation ratio of our algorithm might be possible using other approximations of the resulting sparse dynamic program.

We gave a sequential algorithm for approximate SEP based on bipartite matching with an approximation ratio which is tight for $k\geq 3$. A $\mathrm{NC}$ algorithm for bipartite matching would result in a parallel algorithm with approximation ratio $2$.
Generalizing our method to REP remains open. We also fixed an existing proof of an approximation randomized rounding algorithm.

%%%%%%%%%%%%%%%%
\bibliographystyle{elsarticle-num} 
\bibliography{refs}
\appendix

\section{More on The MPC Model}
Since data are in the form of $(key,value)$ pairs, the amount of communication for each machine in each round is bounded by its memory. The complexity of a MapReduce algorithm is measured in the number of rounds and its total communication, which is at most $O(mLt)$. The number of computations in each machine in each round must be polynomial.
The replication factor of a MapReduce algorithm is the ratio between the maximum used memory in a round divided by the input size.

{\em Sorting in MPC} takes $O(\log_m n)$ rounds~\cite{goodrich2011sorting}, which in MPC and MRC is a constant $c'$. The definition of sorting in MapReduce is to partition data between machines such that the maximum of the numbers in the $i$-th machine is less than or equal to the minimum of the numbers in the $j$-th machine if $i\leq j$. It also gives a balanced partitioning of the data.

{\em Semigroup} operations e.g. minimum, maximum and summation of $n$ numbers, take $O(\log_m n)$ rounds in MapReduce~\cite{goodrich2011sorting}.
{\em Parallel prefix} computations are semigroup computations of prefixes of $n$ numbers, e.g. $x_1+x_2+\cdots+x_i$, for $i=1,\ldots,n$, and take $O(\log_m n)$ rounds and $O(n\log_m n)$ communications in MapReduce~\cite{goodrich2011sorting}.
{\em Sending data in MapReduce} from a machine to a set of machines can be modeled as a parallel prefix computation with the identity operator. So, it takes $O(\log_m n)$ rounds and $O(k\log_m n)$ communications, assuming the amount of data to be sent is $O(k)$. Since it can be seen as a set of $k$ independent parallel prefix computations or a computation that carries data of size $k$.

Two main challenges of designing parallel algorithms with super-constant memory machines are partitioning the data among the machines and partitioning the processing into a set of synchronous rounds. In parallel algorithms, the length of the longest path in the dependency graph of the processes (a DAG) is a lower bound on the time complexity. Adapting this method to massively parallel algorithms requires grouping together the consecutive vertices, such that the length of the longest path becomes small enough.

\subsection{Linear Programming in MPC}
Since linear programming is $P$-complete~\cite{dobkin1979linear} even when the coefficients matrix has only non-negative entries~\cite{trevisan1998parallel}, it has no polylogarithmic time parallel algorithm, assuming $NC\neq P$. The $P$-completeness of linear programming also holds for constant factor approximations~\cite{serna1991approximating,greenlaw1995limits}.
In MRC and MPC models~\cite{mrc,mpc}, the linear programming has only been solved in fixed dimensions using $O(1)$ rounds~\cite{goodrich2011sortingI}.

\section{Dynamic Programs}
Dynamic programming is an important algorithm design tool based on recurrence relations that decreases the need for repeated computations by storing the results of previous computations (memoization) and computes the terms of the recurrence relation in a specific order that allows for better time complexity.

\Cref{fig:fill} shows several filling orders for dynamic programs. Here, $i\oplus j$ is the concatenation of the binary representations of $i$ and $j$, i.e., the bits of $i$ followed by the bits of $j$. Also, $i\otimes j$ interleaves the bits of $i$ and $j$, i.e., one bit of $i$, then one bit of $j$, and it repeats for the digits in other place values. Most of the orders in the first row of \Cref{fig:fill} are converted to the first one to be processed as blocks and the solution is chosen based on the values of the adjacent cells in one or both of the two directions. The second row shows three orderings that only differ based on their tie-breaking rules and the value of each cell depends on a subset of all the $4$ possible directions.

\begin{figure}[h]
\centering
\includegraphics[scale=0.95]{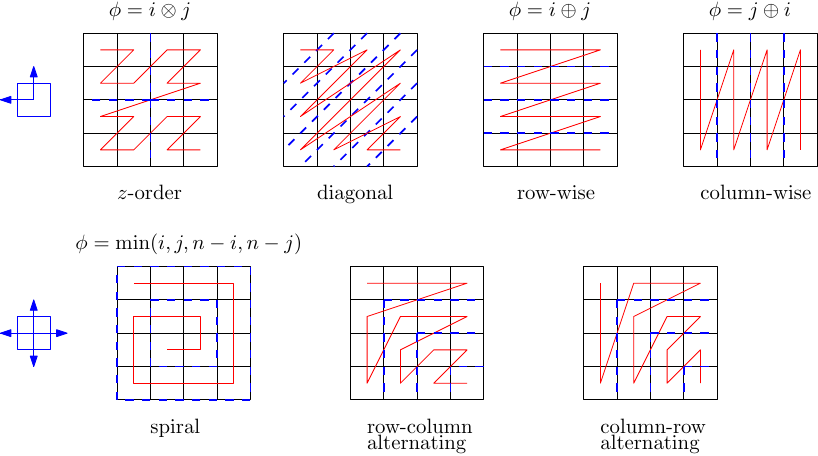}
\caption{Some of the orderings of filling dynamic programming tables.}\label{fig:fill}
\end{figure}

\subsection{Literature Review of The Rectangle Escpae Problem}\label{app:rep}
PCB routing problem with the application of routing rectangular chips in printed circuit boards (PCBs) has been extensively studied before~\cite{routing,obstacle,ma,kong,zar1,boundary,fixed,roy2016runaway}. A printed circuit board is a rectangular grid of pins. Ma, et al~\cite{ma} introduced the rectangle escape problem (REP) to solve the printed circuit board (PCB) bus routing problem, in which the pins of a rectangle should be routed together. They proved REP is NP-hard for $k\geq 3$. Later, Assadi, et al~\cite{zar1} proved that it is NP-hard for $k\geq 2$ and gave a lower bound 3/2 for its approximation factor.

Special cases of this problem have also been studied. Approximation algorithms for maximizing the number of rectangles for a given density, instead of minimizing the density are also known in special cases~\cite{boundary,fixed,chan,chanp,erl}.
An integrality gap of $1.77$ exists for the bidirectional case, where only escaping in the direction of two adjacent boundary edges is allowed~\cite{ahmadinejad2017rectangle}.

\end{document}